\documentclass[10pt]{article}
\usepackage{times}
\usepackage{helvet}
\usepackage{courier}
\usepackage{paralist}
\usepackage{latexsym}
\usepackage{url}
\usepackage[all]{xy}
\usepackage{amsmath}
\usepackage{amssymb}
\usepackage{comment}
\usepackage{xcolor}
\usepackage{times}
\usepackage{graphicx}
\usepackage{savesym}
\savesymbol{AND}
\usepackage{xspace}
\usepackage{tikz}
\usepackage{pgfplots}
\usepackage{pgf}
\usepackage{natbib}
\usepackage{algorithm}
\usepackage{algorithmic}
\usepackage{xspace}
\usepackage{comment}

\usepackage{ssltr}
\usepackage{macros}
\usetikzlibrary{intersections}
\usetikzlibrary{arrows,calc,fit,patterns,plotmarks,shapes.geometric,shapes.misc,shapes.symbols,   shapes.arrows,   shapes.callouts,   shapes.multipart,   shapes.gates.logic.US,   shapes.gates.logic.IEC,   er,   automata,   backgrounds,   chains,   topaths,   trees,   petri,   mindmap,   matrix,   calendar,   folding, fadings,   through,   positioning,   scopes,   decorations.fractals,   decorations.shapes,   decorations.text,   decorations.pathmorphing,   decorations.pathreplacing,   decorations.footprints,   decorations.markings, shadows,circuits}

\newcommand{\nn}{\nonumber}
\newcommand{\nd}{n\times d}
\newcommand{\minp}{(\min,+)}
\newcommand{\maxp}{(\max,+)}
\newcommand{\R}{\mathbf{R}}
\newcommand{\Rm}{\mathbf{R}_{\min}}
\newcommand{\ra}{\rightarrow}
\newcommand{\om}{\otimes}
\newcommand{\op}{\oplus}

\newcommand{\V}{\mathcal{V}}

\newenvironment{proof}{{\bf Proof:} }{}
\newtheorem{theorem}{Theorem}
\newtheorem{lemma}[theorem]{Lemma}

\newtheorem{definition}[theorem]{Definition}

\newtheorem{corollary}{Corollary}

\newtheorem{example}{Example}

\def\R{\mathrm{R}}

\def\S{\mathcal{S}}

\newcounter{subequation}[equation]

\def\mathdisplay#1{%
  \ifmmode \@badmath
  \else
    $$\def\@currenvir{#1}%
    \let\dspbrk@context\z@
    \let\tag\tag@in@display \SK@equationtrue %\let\label\label@in@display
    \global\let\df@label\@empty \global\let\df@tag\@empty
    \global\tag@false
    \let\mathdisplay@push\mathdisplay@@push
    \let\mathdisplay@pop\mathdisplay@@pop
    \if@fleqn
      \edef\restore@hfuzz{\hfuzz\the\hfuzz\relax}%
      \hfuzz\maxdimen
      \setbox\z@\hbox to\displaywidth\bgroup
        \let\split@warning\relax \restore@hfuzz
        \everymath\@emptytoks \m@th $\displaystyle
    \fi
}

\newcounter{algostep}

\newcounter{acalgorithm}

\trinfo{2014}{3}{17}
\title{Approximate dynamic programming with $\minp$ linear function approximation for Markov decision processes.
}
\author{Chandrashekar L$^\dag$ \and Shalabh Bhatnagar$^\$$}
\affiliation{$^\dag$Dept Of CSA, IISc {\tt chandrul@csa.iisc.ernet.in}, $^\$$Dept Of CSA, IISc {\tt shalabh@csa.iisc.ernet.in}}
\everymath{\displaystyle}
\begin{document}
%\maketrcover
\maketitle

\begin{abstract}
Markov Decision Processes (MDP) is an useful framework to cast optimal sequential decision making problems. Given any MDP the aim is to find the optimal action selection mechanism i.e., the optimal policy. Typically, the optimal policy ($u^*$) is obtained by substituting the optimal value-function ($J^*$) in the Bellman equation. Alternately $u^*$ is also obtained by learning the optimal state-action value function $Q^*$ known as the $Q$ value-function. However, it is difficult to compute the exact values of $J^*$ or $Q^*$ for MDPs with large number of states. Approximate Dynamic Programming (ADP) methods address this difficulty by computing lower dimensional approximations of $J^*$/$Q^*$. Most ADP methods employ linear function approximation (LFA), i.e., the approximate solution lies in a subspace spanned by a family of pre-selected basis functions. The approximation is obtain via a linear least squares projection of higher dimensional quantities and the $L_2$ norm plays an important role in convergence and error analysis. In this paper, we discuss ADP methods for MDPs based on LFAs in $\minp$ algebra. Here the approximate solution is a $\minp$ linear combination of a set of basis functions whose span constitutes a subsemimodule. Approximation is obtained via a projection operator onto the subsemimodule which is different from linear least squares projection used in ADP methods based on conventional LFAs. MDPs are not $\minp$ linear systems, nevertheless, we show that the monotonicity property of the projection operator helps us to establish the convergence of our ADP schemes. We also discuss future directions in ADP methods for MDPs based on the $\minp$ LFAs.
\end{abstract}

\section{Introduction}
Optimal sequential decision making problems in science, engineering and economics can be cast in the framework of Markov Decision Processes (MDP). Given an MDP, it is of interest to compute the optimal value-function ($J^*\in \R^n$) and/or the optimal-policy($u^*$), or $Q^* \in \R^{\nd}$ known as the $Q$ value-function which encodes both $J^*$ and $u^*$. The Bellman operator and Bellman equation play a central role in computing optimal value-function ($J^*\in \R^n$) and optimal policy ($u^*$). In particular, $J^*=TJ^*$ and $Q^*=HQ^*$, where $T\colon \R^n \ra \R^n $, $H\colon \R^{n\times d} \ra \R^{n\times d}$ are the Bellman and $Q$-Bellman operators respectively. Most methods to solve MDP such as value/policy iteration (\cite{BertB}) exploit the fact that $J^*$ and $Q^*$ are fixed points of the $T$ and $H$.\\
\indent Most problems arising in practice have large number of states and it is expensive to compute exact values of $J^*/Q^*$ and $u^*$. A practical way to tackle the issue is by resorting to approximate methods. Approximate Dynamic Programming (ADP) refers to an entire spectrum of methods that aim to obtain approximate value-functions and/or policies. In most cases, ADP methods consider a family of functions and pick a function that approximates the value function well. Usually, the family of functions considered is the linear span of a set of basis functions. This is known as linear function approximation (LFA) wherein the value-function of a policy $u$ is approximated as $J_u\approx\tilde{J}_u=\Phi r^*$. Here $\Phi$ is an $n \times k$ $\emph{feature}$ matrix and $r^* \in \R^k$ ($k<<n$) is the weight vector to be computed.\\
Given $\Phi$, ADP methods vary in the way they compute $r^*$. In this paper, we focus on ADP methods that solve the following Projected Bellman equation (PBE),
\begin{align}\label{pbe}
\Phi r^*= \Pi T_u \Phi r^*,
\end{align} 
where $\Pi$ is the projection matrix, $\Pi=\Phi(\Phi^\top D \Phi )^{-1} \Phi^\top$ and $D$ is any positive definite matrix. Once can show that \eqref{pbe} has a solution by showing that the projected Bellman operator (PBO) $\Pi T_u$ is a contraction map in the $L_2$ norm. Solving \eqref{pbe} only address the problem of policy-$\emph{evaluation/prediction}$, and to address the problem of $\emph{control}$ a policy improvement step is required. In order to guarantee an improvement in the policy, the prediction error $||J_u-\tilde{J}_u||_\infty$ needs to be bounded. Due to the use of linear least squares projection operator $\Pi$, we can bound only $||J_u-\tilde{J}_u||_D$, where $||x||_D=x^\top D x$. Consequently policy improvement is not guaranteed and an approximate policy iteration scheme will fail to produce a convergent sequence of policies. Thus the problems of prediction and control are addressed only partially. In particular, convergence and performance bounds are unsatisfactory. Also there is no convergent scheme using conventional LFAs that can approximate $Q^*$.\\
\indent The $\minp$ algebra differs from conventional algebra, in that $+$ and $\times$ operators are replaced by $\min$ and $+$ respectively. Similarly $\maxp$ algebra replaces $+$ with $\max$ and $\times$ with $+$. It is known that finite horizon deterministic optimal control problems with reward criterion are $\maxp$ linear transformations which map the cost function to the optimal-value function. A lot of work has been done in literature \cite{akian,Gaubert,mc2009,gaubert2011} that make use of $\maxp$ basis to solve to compute approximate value-functions.
However, in the case of infinite horizon discounted reward MDP, due to the stochastic nature of the state evolution the Bellman operator ($T$/$H$) is neither $\maxp$ nor $\minp$ linear, which is a key difference from the aforementioned works that apply $\maxp$ basis. The primary focus of the paper it so explore $\minp$ LFAs in ADP schemes as opposed to conventional LFAs. Our specific contributions in this paper are listed below.
\begin{enumerate}
\item We argue that the Bellman operator arising in MDPs are neither $\maxp$/$\minp$ linear. We justify our choice of $\minp$ linear basis for value-function approximation in infinite horizon discounted reward MDPs.
\item We show that the projected Bellman operator in the $\minp$ basis is a contraction map in the $L_\infty$ norm. This enables us establish convergence and error bounds for the ADP schemes developed in this paper.
\item We develop a convergent ADP scheme called Approximate $Q$ Iteration (AQI) to compute $\tilde{Q}=Q^*$. Thus we solve both $\emph{prediction}$ and $\emph{control}$ problems which was a shortcoming in ADP with conventional LFAs.
\item We also present another convergent ADP scheme called Variational Approximate $Q$ Iteration (VAQI), which is based on the Variational or weak formulation of the PBE.
\item We present the error analysis of AQI and VAQI.
\end{enumerate}
Since the main focus of this paper was to study ADP in $\minp$ LFAs and the properties of the associated PBE, we have left out details on algorithmic implementation and analysis of the computational efficiency. Nevertheless, we present experimental results on random MDPs.
\vspace{-2pt}

\section{Markov Decision Processes}\label{mdp}
Markov Decision Processes (MDP) are characterized by their state space $S$, action space $A$, the reward function $g\colon S\times A \ra \R$, and the probability of transition from state $s$ to $s'$ under action $a$ denoted by $p_a(s,s')$. The reward for selecting an action $a$ in state $s$ is denoted by $g_a(s)$. We consider MDP with state space $S=\{1,2,\ldots,n\}$ and action set $A=\{1,2,\ldots,d\}$. For simplicity, we assume that all actions $a \in A$ are feasible in every state $s \in S$. A policy is a map $u\colon S\ra A$, and it describes the action selection mechanism\footnote{The policy thus defined are known as stationary deterministic policy (SDP). Policies can also be non-stationary and randomized. However since there exists a optimal policy that is SDP (\cite{Puter}) we restrict our treatment to SDPs.}. Under a policy $u$ the MDP is a Markov chain and we denote its probability transition kernel by $P_u=(p_{u(i)}(i,j),i=1\mbox{ to }n, j=1\mbox{ to }n)$. The discounted reward starting from state $s$ following policy $u$ is denoted by $J_u(s)$ and is defined as
\begin{align}\label{disc}
J_u(s)=\mathbf{E}[\sum^\infty_{t=0} \alpha^t g_{a_t}(s_t)|s_0=s,u].
\end{align}
Here $\{s_t\}$ is the trajectory of the Markov chain under $u$, and $a_t=u(s_t),\forall t\geq 0$. We call $J_u=(J_u(s),\forall s \in S) \in \R^n$ the value-function for policy $u$. The optimal policy denoted as $u^*$ is given by 
\begin{align}\label{optpol}
u^*=\arg\max_{u \in U}  J_u(s), \forall s \in S.
\end{align}
The optimal value function is then $J^*(s)=J_{u^*}(s), \forall s \in S$. The optimal value function and optimal policy are related by the Bellman Equation as below:
\begin{subequations}\label{bell}
\begin{align}
\label{bellval}J^*(s)&= \max_{a \in A} (g_a(s)+\alpha\sum^n_{s'=1}p_a(s,s') J^*(s')),\\
\label{bellpol}u^*(s)&= \arg\max_{a \in A} (g_a(s) +\alpha\sum^n_{s'=1}p_a(s,s')J^*(s')).
\end{align}
\end{subequations}
Usually, $J^*$ computed is first and $u^*$ is obtained by substituting $J^*$ in \eqref{bellpol}. One can also define the state-action value-function of a policy $u$ known as the $Q$ value-function as follows:
\begin{align}\label{qdef}
Q_u(s,a)=\mathbf{E}[\sum^\infty_{t=0} \alpha^t g_{a_t}(s_t)|s_0=s,a_0=a], a_t=u(s_t) \forall t>0.
\end{align}
The optimal $Q$ values obeys the $Q$ Bellman equation given below:
\begin{align}\label{qopt}
Q^*(s,a)=g_a(s)+\alpha\sum_{s'}p(s,s')\max_{a'}Q^*(s',a').
\end{align}
It is also known that (\cite{BertB}) $J^*(s)=\max_a Q^*(s,a)$. The optimal policy can be computed as $u^*(s)=\arg\max_a Q^*(s,a)$. Thus, in some cases it is beneficial to find $Q^*$ since it encodes both $J^*$ and $u^*$.
\subsection{Basic Solution Methods}
It is important to note that $J^*$ and $Q^*$ are fixed points of maps $T \colon \mathbf{R}^n \rightarrow \mathbf{R}^n$, $H \colon \mathbf{R}^{n\times d} \rightarrow \mathbf{R}^{n\times d}$ respectively defined as follows:
%We define the Bellman and $Q$ Bellman operators $T \colon \mathbf{R}^n \rightarrow \mathbf{R}^n$, $H \colon \mathbf{R}^{n\times d} \rightarrow \mathbf{R}^{n\times d}$ as
\begin{subequations}
\begin{align}
\label{bellop}(TJ)(s)&=\max_{a\in A}( g_a(s)+\alpha \sum^n_{j=1}p_a(s,s')J(s')), J \in \mathbf{R}^n.\\
\label{qbellop}(HQ)(s,a)&=( g_a(s)+\alpha \sum^n_{j=1}p_a(s,s')\max_{a\in A}Q(s',a)), Q \in \mathbf{R}^{n\times d}.
\end{align}
\end{subequations}
$T$ and $H$ are called the Bellman and $Q$-Bellman operators respectively. Given $J \in \R^n, Q\in \R^{\nd}$, $TJ$ and $HQ$ are the `one-step' greedy value-functions. We summarize certain useful properties of $H$ in the following Lemmas (see \cite{BertB} for proofs). 
\begin{lemma}\label{maxnorm}
$H$ is a $\max$-norm contraction operator, i.e., given $Q_1, Q_2 \in \mathbf{R}^{\nd}$ 
\begin{align}
||HQ_1-HQ_2||_\infty\leq \alpha ||Q_1-Q_2||_\infty
\end{align}
\end{lemma}
\begin{corollary}\label{uni}
$Q^*$ is a unique fixed point of $H$.
\end{corollary}
%Further, $H$ exhibits two more important properties presented in the following Lemmas 
\begin{lemma}\label{monotoneH}
$H$ is a monotone map, i.e., given $Q_1,Q_2 \in \R^{\nd}{\nd}$ is such that $Q\geq HQ$, it follows that $Q\geq Q^*$.
\end{lemma}
\begin{lemma}\label{shiftH}
Given $Q\in \R^{\nd}$, and $k \in \R$ and $\mathbf{1} \in \R^{\nd}$ a vector with all entries $1$, then 
\begin{align}
H(Q+k\mathbf{1})=HQ+\alpha k\mathbf{1}.
\end{align}
\end{lemma}
It is easy to check that Lemmas\ref{uni}, \ref{monotoneH}, \ref{shiftH} hold for $T$ as well (\cite{BertB}).\\
Value iteration (VI) is the most basic method to compute $J^*/Q^*$ and is given by
\begin{subequations}\label{iters}
\begin{align}
\label{valiter}J_{n+1}&=TJ_n,\\
\label{qvaliter}Q_{n+1}&=HQ_n.
\end{align}
\end{subequations}
Iterations in \eqref{iters} are exact methods, and the contraction property of the Bellman operator ensures that $J_n\ra J^*$ in \eqref{valiter}, $Q_n\ra Q^*$ in \eqref{qvaliter} as $n\ra\infty$. They are also referred to as $\emph{look-up-table}$ methods or full state representation methods, as opposed to methods employing function approximation. $u^*$ can be computed by substituting $J^*$ in \eqref{bellpol}. Another basic solution method is Policy Iteration (PI) presented in Algorithm~\ref{politer}.
\begin{algorithm}
\caption{Policy Iteration}
\begin{algorithmic}[1]
\STATE Start with any policy $u_0$
\FOR{$i=0,1,\ldots,n$} 
\STATE Evaluate policy $u_{i}$ by computing $J_{u_i}$.
\STATE Improve policy $u_{i+1}(s)=\arg\max_a (g_a(s)+\alpha\sum_{s'}p_a(s,s')J_{u_i}(s'))$.
\ENDFOR
\RETURN $u_n$
\end{algorithmic}
\label{politer}
\end{algorithm}
VI and PI form the basis of ADP methods explained in the next section.

\section{Approximate Dynamic Programming in conventional LFAs}
The phenomenon called $\emph{curse-of-dimensionality}$ denotes the fact that the number of states grows exponentially in the number of state variables. Due to the $\emph{curse}$, as the number of variables increase, it is hard to compute exact values of $J^*/Q^*$ and $u^*$. Approximate Dynamic Programming (ADP) methods make use of \eqref{bell} and dimensionality reduction techniques to compute an approximate value-function $\tilde{J}$. Then $\tilde{J}$ can be used to obtain an approximate policy $\tilde{u}$ which is greedy with respect to $\tilde{J}$ as follows:
\begin{align}\label{subpol}
\tilde{u}(s)=\arg\max_a(g_a(s)+\alpha\sum_{s'}p_a(s,s')\tilde{J}(s')).
\end{align}
The sub-optimality of the greedy policy is given by the following result.
\begin{lemma}\label{subopt}
Let $\tilde{J}=\Phi r^*$ be the approximate value function and $\tilde{u}$ be as in \eqref{subpol}, then 
\begin{align}
||J_{\tilde{u}}-J^*||_\infty \leq \frac{2}{1-\alpha}||J^*-\tilde{J}||_\infty
\end{align}
\end{lemma}
\begin{proof}
See \cite{BertB}.
\end{proof}
Thus a good ADP method is one that address both $\emph{prediction}$ (i.e., computing $\tilde{J}$) and the $\emph{control}$ (i.e., computing $\tilde{u}$) problems with desirable approximation guarantees.\\
Linear function approximators (LFA) have been widely employed for their simplicity and ease of computation. LFAs typically let $\tilde{J} \in V \subset\R^n$, where $V$ is the subspace spanned by a set of preselected basis functions $\{\phi_i \in \R^n, i=1,\ldots,k\}$. 
Let $\Phi$ be the $n\times k$ matrix with columns $\{\phi_i\}$, and $V=\{\Phi r|r\in \R^k\}$, then the approximate value function $\tilde{J}$ is of the form $\tilde{J}=\Phi r^*$ for some $r^* \in \R^k$. $r^*$ is a weight vector to be learnt, and due to dimensionality reduction ($k<<n$) computing $r^* \in \R^k$ is easier than computing $J^* \in \R^n$.\\
We now discuss popular ADP methods namely approximate policy evaluation (APE) and approximate policy iteration (API), which are approximation analogues of VI and PI. APE and API are based on linear least squares projection of higher dimensional quantities onto $V$. %It turns out that it is not straightforward to obtain approximations to $J^*$.
\subsection{Approximate Policy Evaluation}
$J^*/Q^*$ are not known and hence projecting them onto $V$ is impossible. Nevertheless, one can use the Bellman operator and linear least squares projection operator to write down a projected Bellman equation (PBE) as below:
\begin{align}\label{projbasic}
\Phi r^*=\Pi T_u\Phi r^*
\end{align}
where $\Pi=\Phi(\Phi^\top D \Phi)^{-1}\Phi^\top$ is the projection operator, $D$ is any diagonal matrix with all entries strictly greater than $0$ and $T_u$ is the Bellman operator restricted to a policy $u$ and is given by 
\begin{align}
(T_uJ)(s)&=g_{u(s)}(s)+\alpha\sum_{s'}p_{u(s)}(s,s')J(s'), J\in \R^n.\nn
\end{align}
The approximate policy evaluation (APE) is the following iteration:
\begin{align}\label{appvaliter}
\Phi r_{n+1}=\Pi T_u \Phi r_{n}.
\end{align}
% Since the projection operator is closely related to the $L_2$ norm, in that $\Pi J=\min_r||J-\Phi||_2$, i.e., it reduces the squared error. However, $T$ is not a contraction operator in the $L_2$ norm. 
%Since the iteration $J_{n+1}=T_uJ_n$, evaluates the policy $u$, i.e., $J_n \ra J_u$ as $n\ra\infty$
The following Lemma~\ref{pbeconv} establishes the convergence of \eqref{appvaliter}.
\begin{lemma}\label{pbeconv}
If $\Pi=\Phi(\Phi^\top D\Phi)^{-1}\Phi^\top$, and $D$ be a diagonal matrix with the $i^{th}$ diagonal entry being the stationary probability of visiting state $i$ under policy $u$, then $\Pi T_u$ is a contraction map with factor $\alpha$.
\end{lemma}
\begin{proof}
Let $P_u$ be the probability transition matrix corresponding to policy $u$. Then one can show that (\cite{BertB}) for $z\in\R^n$ and $||z||^2_D=z^\top Dz$, $||P_uz||_D^2\leq ||z||_D^2$. Also, we know that $\Pi$ is a non-expansive map, because
\begin{align}
||\Pi x-\Pi y||_D&=||\Pi (x-y)||_D \nn\\
&\leq ||\Pi (x- y)||_D+||(I-\Pi) (x-y)||_D\nn\\
&=||x-y||_D\nn
\end{align}
Then
\begin{align}
||\Pi TJ_1-\Pi TJ_2||_D\leq || TJ_1-TJ_2||_D\leq \alpha ||J_1-J_2||_D\nn
\end{align}
\end{proof}
\begin{corollary}
Then the iteration in \eqref{appvaliter} converges to $\Phi r^*$ such that $\Phi r^*=\Pi T_u\Phi r^*$.
\end{corollary}
The error bound for $\Phi r^*$ is given by
\begin{align}\label{errbnd}
||J_u-\Phi r^*||_D \leq \frac{1}{\sqrt{1-\alpha^2}}||J_u-\Pi J_u||_D,
\end{align}
One is inclined to think that an approximation analogue of \eqref{valiter} would yield an approximation to $J^*$. It is important to note that \eqref{appvaliter} only computes $\tilde{J}_u$ because \eqref{appvaliter} contains $T_u$ and not $T$. However, since operator $\Pi T$ might not be a contraction map in the $L_2$ norm, and $T_u$ cannot be replaced by $T$ in iteration \eqref{appvaliter}, and the PBE in \eqref{projbasic}. 
\subsection{Approximate Policy Iteration}
Approximate Policy Iteration (Algorithm~\ref{appoliter}) tackles both prediction and control problems, by performing APE and policy improvement at each step.
\begin{algorithm}
\caption{Approximate Policy Iteration (API)}
\begin{algorithmic}[1]
\STATE Start with any policy $u_0$
\FOR{$i=0,1,\ldots,n$} 
\STATE \label{appeval}Approximate policy evaluation $\tilde{J}_{i}= \Phi r^*_i$, where $\Phi r^*_i=\Pi T_{u_i}\Phi r^*_i$.
\STATE Improve policy $u_{i+1}(s)=\arg\max_a (g_a(s)+\alpha\sum_{s'}p_a(s,s')\tilde{J}_{i}(s'))$.
\ENDFOR
\RETURN $u_n$
\end{algorithmic}
\label{appoliter}
\end{algorithm}
The performance guarantee of API can be stated as follows:
\begin{lemma}\label{pgapi}
If at each step $i$ one can guarantee that $||\tilde{J}_i-J_{u_i}||_\infty\leq \delta$, then one can show that $\lim_{n\ra\infty}||J_{u_i}-J^*||\leq \frac{2\delta\alpha}{(1-\alpha)^2}$.
\end{lemma}
Note that the error bound required by Lemma~\ref{pgapi} is in the $L_\infty$ norm, whereas \eqref{errbnd} is only in the $L_2$ norm. So API cannot guarantee an approximate policy improvement each step which is a shortcoming. Also, even-though each evaluation step (line~\ref{appeval} of Algorithm~\ref{appoliter}) converges, the sequence ${u_n}, n\geq0$ is not guaranteed to converge. This is known as policy $\emph{chattering}$ and is another important shortcoming of conventional LFAs. Thus the problem of $\emph{control}$ is only partially addressed by API.
\subsection{LFAs for $Q$ value-function}\label{appqlearn}
To alleviate the shortcomings in API, it is then natural to look for an approximation method that computes a policy in more direct fashion. Since we know that by computing $Q^*$ we obtain the optimal policy directly, it is a good idea to approximate $Q^*$. The PBE version of \eqref{qvaliter} is a plausible candidate and the iterations are given by,
\begin{align}\label{appqiter}
\Phi r_{n+1}=\Pi H \Phi r_n.
\end{align}
The above scheme will run into the following problems:
\begin{enumerate}
\item The $H$ operator \eqref{qbellop} contains a $\max$ term, and it is not straightforward show that $\Pi H$ is a contraction map in $L_2$ norm, and consequently one cannot establish the convergence of iterates in \eqref{appqiter}.
\item The issue pertaining to $\max$ operator can be alleviated by restricting $H$ to a policy $u$, i.e., consider iterations of form 
\begin{align}\label{resappqiter}
\Phi r_{n+1}=\Pi H_u \Phi r_{n+1}.
\end{align}
But iterates in \eqref{resappqiter} attempt to approximate $Q_u$ and not $Q^*$, which means the problem of $\emph{control}$ is unaddressed.
\end{enumerate}
We conclude this section with the observation that the important shortcomings of conventional LFAs related to convergence and error bound arise due to the $L_2$ norm. The main result of the paper is that ADP methods based on $\minp$ LFAs don't suffer from such shortcomings.

\section{$\minp$/$\maxp$ non-linearity of MDPs}
We introduce the $\Rm$ semiring and show that MDPs are neither $\minp$ nor $\maxp$ linear. The $\Rm$ semiring is obtained by replacing multiplication ($\times$) by $+$, and addition ($+$) by $\min$.
\begin{definition}
\begin{align}
&\text{Addition:} &x \op y&= \min(x,y)\nn\\
&\text{Multiplication:} &x \om y&= x+y\nn
\end{align}
\end{definition}
Henceforth we use, $(+, \times)$ and $(\op,\om)$ to respectively denote the conventional and $\Rm$ addition and multiplication respectively.
%We later justify the choice of $\minp$ linear basis for the discounted reward MDP setting.\\
Semimodule over a semiring can be defined in a similar manner to vector spaces over fields. In particular we are interested in the semimodule $\mathcal{M}=\Rm^n$. Given $u, v \in \Rm^n$, and $\lambda \in \Rm$, we define addition and scalar multiplication as follows:
\begin{definition}
\begin{align}
(u\op v)(i)&=\min\{u(i),v(i)\}=u(i)\op v(i), \forall i=1,2,\ldots,n.\nn\\
(u\om \lambda)(i)&=u(i)\om \lambda=u(i)+\lambda, \forall i=1,2,\ldots,n.\nn
\end{align}
\end{definition}
Similarly one can define the $\R_{\max}$ semiring which has the operators $\max$ as addition and $+$ as multiplication.\\
\indent It is a well known fact that deterministic optimal control problems with cost/reward criterion are $\minp$/$\maxp$ linear. However,
the Bellman operator $T$ in \eqref{bellop} (as well as $H$ in \eqref{qbellop}) corresponding to infinite horizon discounted reward MDPs are neither $\minp$ linear nor $\maxp$ linear systems. We illustrate this fact via the following example.
\begin{example}
Consider an MDP with two states $S=\{s_1,s_2\}$, only one action, and a reward function $g$, and let the probability transition kernel be
\begin{align}\label{probtran}
P=\left[ {\begin{array}{cc}  0.5 & 0.5   \\ 0.5 & 0.5\\ \end{array} } \right]
\end{align}
For any $J\in \R^2$ the Bellman operator $T\colon \R^2 \ra \R^2$ can be written as
\begin{align}
(TJ)(s)=g(s)+\alpha\times(0.5\times J(1)+0.5\times J(2))
\end{align}
Consider vectors $J_1,J_2 \in \R^2$ such that $J_1=(1,2)$ and $J_2=(2,1)$ and $J_3=\max(J_1,J_2)=(2,2)$. Let $g(1)=g(2)=1$, and $\alpha=0.9$, then, it is easy to check that $TJ_1=TJ_2=(2.35,2.35)$, and $TJ_3=(2.8,2.8)$. However $TJ_3\neq \max(TJ_1,TJ_2)$, i.e., $TJ_3\neq (2.35,2.35)$. Similarly one can show that $T$ is neither a $\minp$ linear operator.
\end{example}$\blacksquare$
\begin{comment}
Consider a finite horizon deterministic control problem, with state space $S=\{1,\ldots,n\}$, and action space $A=\{1,\ldots,d\}$ and horizon $T$. An action $a$ in state $s$, takes the system to $s(a)$. Let the cost be $r\in \Rm^n$ and $(J_t(s),\forall s\in S)\in \Rm^n$ denote the corresponding optimal $\emph{cost-to-go}$ function at horizon $t$. The cost-to-go functions then obey the Bellman equation given below:
\begin{align}
J_t=\min_{a \in A}(r(s)+J_{t+1}(s(a))), \forall t=0,\ldots,T-1.
\end{align}
We define operator $\mathcal{S} \colon \Rm^n\ra \Rm^n$ that maps the cost function to the corresponding cost-to-go function at time $t=0$, i.e., $\mathcal{S}(r)=J_0$. The operator is $\minp$ linear, i.e., for cost functions $u_1,u_2\in \Rm^n$
\begin{align}
\S(u_1\op u_2)&=\S(u_1)\op\S(u_2),\\ %\mbox{ }\text{for}\mbox{ }\text{cost functions}\mbox{ }u_1,u_2.
\S(u_1\op k\times e)&=\S (u_1)\om k\times e,\mbox{ } \text{where}\mbox{ } k \in \R,
\end{align}
and $e\in \Rm^n$ has all its components to be $1$. Similarly, one can show that deterministic control problems with reward criterion are $\maxp$ linear. However, in this paper, we make use of $\minp$ basis for infinite horizon discounted reward MDPs. We justify this apparently counter-intuitive choice in the next section.
\end{comment}

\section{$\minp$ linear functions}\label{semiring}
Even-though the Bellman operator is not $\minp$/$\maxp$ linear, the motivation behind developing ADP methods based on $\minp$ LFAs is to explore them as an alternative to the conventional basis representation. Thus the aim to understand the kind of convergence guarantees and error bounds that are possible in the $\minp$ LFAs.\\
Given a set of basis function $\{\phi_i,i=1,\ldots,k\}$, we define its $\minp$ linear span to be $\mathcal{V}=\{v|v=\Phi\om r\stackrel{def}{=}\min(\phi_1+r(1),\ldots,\phi_k+r(k)), r \in \Rm^k\}$. $\mathcal{V}$ is a subsemimodule. In the context value function approximation, we would want to project quantities in $\Rm^n$ onto $\mathcal{V}$. The $\minp$ projection operator $\Pi_M$ is given by (\cite{akian,cohen1996kernels,Gaubert})
\begin{align}\label{smproj}
\Pi_M u=\min\{ v | {v \in \mathcal{V}}, v \geq u\}, \forall u \in \mathcal{M}.
\end{align}
We can write the PBE in the $\minp$ basis
\begin{align}\label{projminpbasic}
v&=\Pi_M Tv, v\in \V\nn\\
\Phi \om r^*&=\min\{\Phi \om r^*\in \V|\Phi\om r^*\geq T \Phi \om r^*\}
\end{align}
Our choice is justified by the similarity in the structure of \eqref{projminpbasic} and the linear programming (LP) formulation of the infinite horizon discounted reward MDP.
\begin{align}\label{lp}
 &\mbox{ }\min ~c^\top J \\
      &\quad s.t\quad J(s)\geq g(s,a)+\alpha\sum_{s'}p_a(s,s')J(s'),\forall s \in S, a \in A\nn
\end{align}
Thus by making use of $\minp$ LFAs and the projection operator $\Pi_M$ we aim to find the minimum upper bound to $J^*$/$Q^*$. A closely related ADP method is the approximate linear program (ALP) given by
\begin{align}\label{alp}
 &\mbox{ }\min ~c^\top \Phi r \\
      &\quad s.t\quad \Phi r(s)\geq g(s,a)+\alpha\sum_{s'}p_a(s,s')(\Phi r)(s'),\forall s \in S, a \in A\nn
\end{align}
Though formulations \eqref{projminpbasic} and \eqref{alp} look similar, the former has its search space which is a subsemimodule formed by a $\minp$ linear span, whereas the latter has search space which is the intersection of subspace and set of constraints. The two formulations differ in the algorithmic implementation and performance bounds which we discuss in a longer version of the paper.

\section{ $\minp$ LFAs for $Q$ value approximation}
We now present an ADP scheme bases on solving the PBE in $\minp$ basis to compute approximate $Q^*\approx \tilde{Q}$. Our ADP scheme successfully addresses the two shortcomings of the ADP scheme in conventional basis. First, we show establish contraction of the projected Bellman operator in the $L_\infty$ norm. This enables us to show that our recursion to compute $\tilde{Q}$ converges. As discussed earlier, we can also obtain a greedy policy $\tilde{u}(s)=\max_a\tilde{Q}(s,a)$. Secondly, we also present an error bound for the $\tilde{Q}$ in the $\max$ norm and as a consequence we can also ascertain the performance of $\tilde{u}$.\\
The PBE we are interested in solving is\footnote{We did not consider $\Phi \om r^*=\Pi_M T \Phi \om r^*$ since it approximates only $J^*$ and is superseded by \eqref{qpbe} which computes approximate $Q^*\approx\tilde{Q}$. Thus \eqref{qpbe} addresses both $\emph{prediction}$ and $\emph{ rol}$ problems. We wish the remind the reader of the issues with \eqref{qpbe} in conventional basis discussed in section~\ref{appqlearn}.}
\begin{align}\label{qpbe}
\Phi \om r^*=\Pi_M H \Phi \om r^* 
\end{align}
Since we want to approximate $Q^*$, $\Phi$ is a $nd\times k$ feature matrix, and $Q^*\approx\tilde{Q}(s,a)=\phi^{(s-1)\times d+a}\om r^*$, where $\phi^i$ is the $i^{th}$ row of $\Phi$.
The projected $Q$ iteration is given by
\begin{align}\label{projqiter}
\Phi \om r_{n+1}=\Pi_M H \Phi \om r_n.
\end{align}
The following results help us to establish the fact that that the operator $\Pi_M H\colon \Rm^{\nd}\ra \Rm^{\nd}$ is a contraction map in the $L_\infty$ norm.
\begin{lemma}\label{monotone}
For $Q_1, Q_2 \in \Rm^{\nd}$, such that $Q_1\geq Q_2$, then $\Pi_M Q_1\geq \Pi_M Q_2$.
\end{lemma}
\begin{proof}
Follows from definition of projection operator in \eqref{smproj}.
\end{proof}
\begin{lemma}\label{shift}
Let $Q\in \Rm^{\nd}$, $V_1=\Pi_M Q$ be its projection onto $\V$ and $k\in \R$, and $\mathbf{1}\in \R^{\nd}$ be a vector with all components equal to $1$. The projection of $Q+k\mathbf{1}$ in $V_2=\Pi_M Q+k\mathbf{1}$.
\end{lemma}
\begin{proof}
We know that since $V_1\geq Q$, $V_1+k\mathbf{1}\geq Q+k\mathbf{1}$, and from the definition of the projection operator $V_2\leq V_1+k\mathbf{1}$. Similarly $V_2-k\mathbf{1}\geq Q$, so $V_1\leq V_2-k\mathbf{1}$.
\end{proof}
\begin{theorem}\label{contra}
$\Pi_M H$ is a contraction map in $L_\infty$ norm with factor $\alpha$.
\end{theorem}
\begin{proof}
Let $Q_1,Q_2\in\Rm^{\nd}$, define $\epsilon\stackrel{def}{=}||Q_1-Q_2||_\infty$, then
\begin{align}
\Pi_M H Q_1- \Pi_M HQ_2&\leq \Pi_M H (Q_2+\epsilon\mathbf{1})-\Pi_M HQ_2\nn\\
				&=\Pi_M (H Q_2 +\alpha\epsilon\mathbf{1})-\Pi_M HQ_2\nn\\
				&=\alpha\epsilon\mathbf{1}.
\end{align}
Similarly $\Pi_M H Q_2- \Pi_M HQ_1\leq \alpha\epsilon\mathbf{1}$, so it follows that $||\Pi_M HQ_1-\Pi_MHQ_2||_\infty\leq \alpha ||Q_1-Q_2||_\infty$.
\end{proof}
\begin{corollary}\label{minpaviconv}
The approximate $Q$ iteration in \eqref{projqiter} converges to a fixed point $r^*$.
\end{corollary}

\section{Variational formulation of PBE}
The projection operator $\Pi_M$ used in \eqref{appqiter} is exact. Let $v=\Pi_M u$, then for $\{w_i \in \Rm^n\},i=1,\ldots,m$ it follows from the definition of $\Pi_M$ that
\begin{align}\label{testfn}
w_i^\top v\geq w_i^\top u
\end{align}
where in the $\minp$ algebra the dot product $x^\top y=\min^n_{i=1}(x(i)+y(i))$. Let $W$ denote the $nd\times m$ test matrix whose columns are $w_i$. Now we shall define the approximate projection operator to be
\begin{align}\label{appproj}
\Pi^W_M u=\min\{v \in \mathcal{V}| W^\top v\geq W^\top u \}.
\end{align}
The superscript in $\Pi_M^W$ denotes the test matrix $W$. The iteration to compute approximate $Q$ values using $\Pi^W_M$ is given by
\begin{align}\label{appprojqiter}
\Phi \om r_{n+1}= \Pi^W_M H\Phi \om r_n.
\end{align}
Lemmas~\ref{monotone},~\ref{shift}, and Theorem~\ref{contra} continue to hold if $\Pi_M$ is replaced with $\Pi^W_M$. Thus by Corollary~\ref{minpaviconv}, we know that \eqref{appprojqiter} converges to an unique fixed point $r^*_W$ such that $\Phi \om r^*_W=\Pi^W_M H\Phi \om r^*_W$.
\begin{theorem}\label{eboundapp}
Let $\tilde{r}$ be such that $\tilde{r}=\arg\min_r||Q^*-\Phi\om r||_\infty$. Let $r^*$ be the fixed point of the iterates in \eqref{appprojqiter}, then
\begin{align}
||Q^*-\Phi\om r^*||_\infty&\leq \frac{2}{1+\alpha}(||Q^*-\Phi\om \tilde{r}||_\infty\nn\\&\mbox{  }+||\Phi\om \tilde{r}-\Pi_M^W \Phi\om \tilde{r}||_\infty)
\end{align}
\end{theorem} 
\begin{proof}
Let $\epsilon=||Q^*-\Phi\om\tilde{r}||_\infty$, by contraction property of $H$ (Lemma~\ref{maxnorm}) we know that
\begin{align}
||HQ^*-H\Phi\om \tilde{r}||_\infty\leq \alpha \epsilon. \nn
\end{align}
So have $||\Phi\om\tilde{r}-H\Phi\om\tilde{r}||_\infty\leq (1+\alpha)\epsilon$. Then
\begin{align}
||\Phi\om\tilde{r}-\Pi^W_M H\Phi \om \tilde{r}||_\infty&=||\Phi\om\tilde{r}-\Pi^W_M \Phi \om \tilde{r}||_\infty\nn\\&+||\Pi^W_M\Phi\om\tilde{r}-\Pi^W_M H\Phi \om \tilde{r}||_\infty\nn
\end{align}
Now
\begin{align}
\Pi^W_M \Phi\om\tilde{r}-\Pi^W_M H\Phi \om \tilde{r}&\leq\Pi^W_M \Phi\om\tilde{r}\nn\\&-\Pi^W_M (\Phi \om \tilde{r}-(1+\alpha)\epsilon)\nn\\
&=(1+\alpha)\epsilon\nn
\end{align}
Similarly $\Pi^W_M H\Phi \om \tilde{r}-\Pi^W_M\Phi\om\tilde{r}\leq (1+\alpha)\epsilon$, and hence 
\begin{align}
||\Phi\om\tilde{r}-\Pi^W_M H\Phi \om \tilde{r}||_\infty\leq (1+\alpha)\epsilon+\beta,
\end{align}
where $\beta=||\Phi\om\tilde{r}-\Pi^W_M \Phi \om \tilde{r}||_\infty$. 
Now consider the iterative scheme in \eqref{appprojqiter} with $r_0=\tilde{r}$, and 
\begin{align}
||Q^*-\Phi\om r^*||_\infty&=||Q^*-\Phi\om r_0+\Phi\om r_0\nn\\&-\Phi\om r_1+\ldots-\Phi\om r^*||_\infty\nn\\
&\leq ||Q^*-\Phi\om r_0||_\infty+||\Phi\om r_0-\Phi\om r_1||_\infty\nn\\&+||\Phi\om r_1-\Phi\om r_2||_\infty+\ldots\nn\\
&\leq \epsilon+(1+\alpha)\epsilon+\beta+\alpha((1+\alpha)\epsilon+\beta)+\ldots\nn\\
&=\epsilon(\frac{1+\alpha}{1-\alpha}+1)+\frac{\beta}{1-\alpha}\nn\\
&=\frac{2\epsilon+\beta}{1-\alpha}\nn
\end{align}
\end{proof}
The term $\beta$ in the error bound in Theorem~\ref{eboundapp} is the error due to the usage of $\Pi^W_M$. Thus for solution to \eqref{qpbe} $\beta=0$.

\section{Experiments}\label{exper}
We test our ADP schemes on a randomly generated MDP with $100$ states, i.e., $S=\{1,2,\ldots,100\}$, and action set $A=\{1,\ldots,5\}$. The reward $g_a(s)$ is a random integer between $1$ and $10$, and discount factor $\alpha=0.9$. We now describe feature selection, where $\{\phi_j, j=1,\ldots,k\}, \phi_j \in \Rm^{\nd}$ and $\{\phi^i,i=1,\ldots,n\}, \phi^i \in \Rm^k$ denote the columns and rows respectively of the feature matrix $\Phi$. The feature corresponding to a state-action pair $(s,a)$ is given by $\phi^{(s-1)\times d+a}$. Let $\phi^x,\phi^y$ be features corresponding to state action pairs $(s_x,a_x)$ and $(s_y,a_y)$ respectively, then
\begin{align}\label{dotp}
<\phi^x,\phi^y>= \phi^x(1)\om\phi^y(1)\op\ldots\op\phi^x(k)\om\phi^y(k).
\end{align}
We desire the following in the feature matrix $\Phi$.
\begin{enumerate}
\item Features $\phi^i$ should have unit norm, i.e., $||\phi^i||=<\phi^i,\phi^i>=\mathbf{0}$, since $\mathbf{0}$ is the multiplicative identity in the $\minp$ algebra.
\item For dissimilar states-action pairs, we prefer $<\phi^x,\phi^y>=+\infty$, since $+\infty$ is the additive identity in $\minp$ algebra.
\end{enumerate}
Keeping these in mind, we design the feature matrix $\Phi$ for the random MDP as in \eqref{feature}. For state-action pair $(s,a)$ let $x=(s-1)\times d+a$, then the feature
\begin{align}\label{feature}
%$$
\phi^{x}(i) = \left\{
        \begin{array}{ll}
            0 & : g_a(s) \in [g_{\min}+\frac{(i-1)L}{k},g_{\min}+\frac{(i)L}{k}] \\
       1000 & : g_a(s) \notin [g_{\min}+\frac{(i-1)L}{k},g_{\min}+\frac{(i)L}{k}],
        \end{array}
    \right.\nn\\
%$$
\forall  i=1,\ldots,k.
\end{align}
We use $1000$ in place of $+\infty$. It is easy to verify that $\Phi$ in \eqref{feature} has the enumerated properties. The results are plotted in Figure \ref{Vapp}. Here $J^*$ is the optimal value-function, $\tilde{J}_{EP}(s)=\max_a\tilde{Q}_{EP}(s,a)$, where $\tilde{Q}_{EP}$ is the value obtained via the iterative scheme in \eqref{projqiter} and subscript $EP$ denotes the fact that the projection employed in $\emph{exact}$ ($\Pi_M$). $\tilde{J}_{W}(s)=\max_a\tilde{Q}_W(s,a)$, where $\tilde{Q}_{W}$ is the value obtained via the iterative scheme in \eqref{appprojqiter} and subscript $W$ denotes the fact that the projection employed is $\Pi^W_M$. $u_{EP}(s)=\arg\max_a \tilde{Q}_{EP}(s,a)$ and  $u_{W}(s)=\arg\max_a \tilde{Q}_{W}(s,a)$ are greedy policies and $J_{u_{EP}}$, $J_{u_W}$ are their respective value functions. $\tilde{J}_{u_{arbt}}$ is the value function of an $\emph{arbitrary}$ policy, wherein a fixed arbitrary action is chosen for each state.
\begin{comment}
\begin{align}
\tilde{u}=\underset{a \in A}{\arg\max}\bigg( g(s)+\alpha \sum p_a(s,s')\tilde{J}(s')\bigg), \\\text{where} \mbox{ } \tilde{J}=\Phi\om r_{opt}.\nn
\end{align}
\end{comment}

\begin{comment}
\begin{table}
\resizebox{\columnwidth}{!}{
\begin{tabular}{|c|c|c|c|}\hline
&  &   &  \\
$||J^*-\tilde{J}_{EP}||_\infty$ &$||J^*-\tilde{J}_W||_\infty$ &$||J^*-J_{u_{EP}}||_\infty$ &$||J^*-J_{u_W}||_\infty$ \\\hline
$6.47$  & $6.35$  & $2.61$  & $5.61$  \\ \hline
$||J^*-\tilde{J}_{u_{arbt}}||_\infty$ & & & \\\hline
$40.49$  &  &  &   \\ \hline
\end{tabular}
}
\caption{Error Table}
\label{errtable}
\end{table}
\end{comment}

\begin{figure}
\begin{minipage}{0.45\textwidth}
\begin{tabular}{c}
\begin{tikzpicture}[scale=0.75]
    \begin{axis}[
	xlabel=State,
        ylabel=Discounted Cost,
	ymin=65,
	legend pos= south west,
	title={Optimal and Approximate Value Functions}
]
    \addplot[smooth,mark=.,blue] plot file {V.dat};
    \addplot[smooth,mark=o,red] plot file {V_Q.dat};
    \addplot[smooth,mark=+,black] plot file {V_Qv.dat};
     \legend{$J^*$,$\tilde{J}_{EP}$, $\tilde{J}_{W}$}
    \end{axis}
    \end{tikzpicture}
\end{tabular}
\end{minipage}
\caption{$||J^*-\tilde{J}_{EP}||_\infty=6.47$,$||J^*-\tilde{J}_W||_\infty=6.35$}
\label{Vapp}
\end{figure}
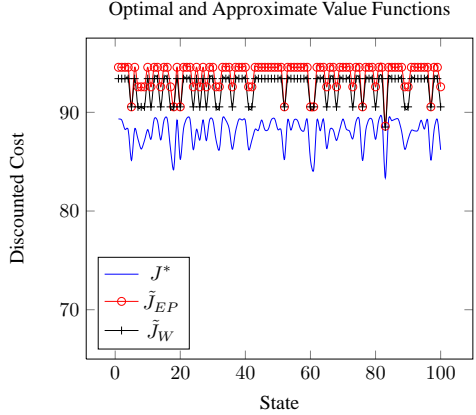

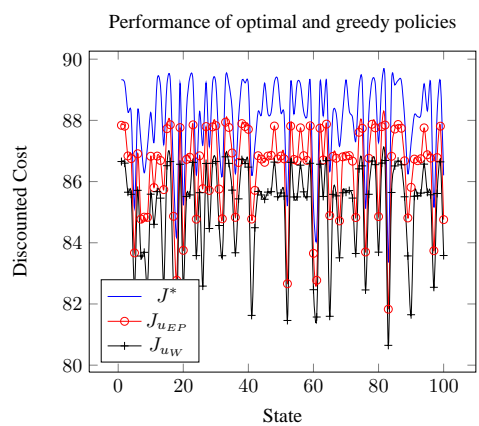
\begin{figure}
\begin{minipage}{0.45\textwidth}
\begin{tabular}{c}
\begin{tikzpicture}[scale=0.75]
    \begin{axis}[
	xlabel=State,
        ylabel=Discounted Cost,
	legend pos=south west,
	title={Performance of optimal and greedy policies}
]
    \addplot[smooth,mark=.,blue] plot file {V.dat};
    \addplot[smooth,mark=o,red] plot file {V_pol.dat};
    \addplot[smooth,mark=+,black] plot file {V_polv.dat};
%    \addplot[smooth,mark=*,green] plot file {V_rand.dat};
     \legend{$J^*$,${J}_{u_{EP}}$, $J_{u_{W}}$}
    \end{axis}
    \end{tikzpicture}
\end{tabular}
\end{minipage}
\caption{$||J^*-J_{u_{EP}}||_\infty=2.61$,$||J^*-J_{u_W}||_\infty=5.61$, $||J^*-\tilde{J}_{u_{arbt}}||_\infty=40.49$}
\label{Vapp}
\end{figure}

\clearpage
\bibliographystyle{plainnat}
\bibliography{ref}
\end{document}